\documentclass{amsart}

\usepackage{amssymb,amsfonts}
\usepackage[all,arc]{xy}
\usepackage{enumerate}
\usepackage{mathrsfs}
\usepackage{tikz}
\usepackage{algorithm}
\usepackage{algorithmicx}
\usepackage{algpseudocode}
\usepackage{float}
\usepackage{hyperref}
\usepackage{wrapfig}

\usetikzlibrary{matrix,arrows}

\newtheorem{thm}{Theorem}[section]

\newtheorem{prop}[thm]{Proposition}

\theoremstyle{definition}
\newtheorem{defn}[thm]{Definition}

\theoremstyle{remark}

\newcommand\CONDITION[2]%
  {\begin{tabular}[t]{@{}l@{}l@{}}
     #1&#2
   \end{tabular}%
  }
\algdef{SE}[WHILE]{While}{EndWhile}[1]%
  {\algorithmicwhile\ \CONDITION{#1}{\ \algorithmicdo}}%
  {\algorithmicend\ \algorithmicwhile}
\algdef{SE}[FOR]{For}{EndFor}[1]%
  {\algorithmicfor\ \CONDITION{#1}{\ \algorithmicdo}}%
  {\algorithmicend\ \algorithmicfor}
\algdef{S}[FOR]{ForAll}[1]%
  {\algorithmicforall\ \CONDITION{#1}{\ \algorithmicdo}}
\algdef{SE}[REPEAT]{Repeat}{Until}{\algorithmicrepeat}[1]%
  {\algorithmicuntil\ \CONDITION{#1}{}}
\algdef{SE}[IF]{If}{EndIf}[1]%
  {\algorithmicif\ \CONDITION{#1}{\ \algorithmicthen}}%
  {\algorithmicend\ \algorithmicif}%
\algdef{C}[IF]{IF}{ElsIf}[1]%
  {\algorithmicelse\ \algorithmicif\ \CONDITION{#1}{\ \algorithmicthen}}
  


\begin{document}

\title[Sums-of-Squares Constraint Propagation]{A Constraint Propagation Algorithm for Sums-of-Squares Formulas over the Integers}
\author{Melissa Lynn}

\begin{abstract}
Sums-of-squares formulas over the integers have been studied extensively using their equivalence to consistently signed intercalate matrices. This representation, combined with combinatorial arguments, has been used to produce sums-of-squares formulas and to show that formulas of certain types cannot exist. In this paper, we introduce an algorithm that produces consistently signed intercalate matrices, or proves their nonexistence, extending previous methods beyond what is computationally feasible by hand.
\end{abstract}

\maketitle

\section{Introduction}

A \emph{sums-of-squares formula of type $[r,s,n]$ over $\mathbb{Z}$} is an identity of the form
\[
(x_1^2+\cdots + x_r^2)(y_1^2+\cdots +y_s^2)=z_1^2+\cdots +z_n^2
\]
where each $z_i$ is a bilinear expression in the $x$'s and $y$'s over $\mathbb{Z}$.

Sums-of-squares formulas have been studied since 1898, when Hurwitz proved that the only real normed division algebras are the real numbers, the complex numbers, the quaternions, and the octonions. This theorem is proved by considering sums-of-squares formulas of type $[n,n,n]$ over $\mathbb{R}$. In his paper, Hurwitz posed the general question: for what types $(r,s,n)$ does a sums-of-squares formula exist over a given field $F$? \cite{Hurwitz1} \cite{Hurwitz2} We may also consider sums-of-squares formulas over rings.

Whether existence of a sums-of-squares formula depends on the base ring or field remains an open question. Sums-of-squares formulas over the integers are particularly important, since a sums-of-squares formula over $\mathbb{Z}$ maps to a formula over any field $F$ via the natural map $\mathbb{Z}\rightarrow F$.

The general question about the existence of sums-of-squares formulas has valuable connections to topology and geometry, since sums-of-squares formulas induce immersions of projective space into Euclidean space, and they induce Hopf maps between spheres. Shapiro's book \cite{ShapiroBook} covers the history of sums-of-squares formulas and past results.

The special case of a sums-of-squares formula over the integers has been studied extensively using combinatorics. Yuzvinsky \cite{Yuzvinsky} introduced an equivalence between sums-of-squares formulas over the integers and consistently signed intercalate matrices. By studying these matrices, Yuzhvinsky and others were able to produce many new formulas, and prove many new results on the existence of formulas of various types.

In this paper, we revisit the equivalence between sums-of-squares formulas over the integers and consistently signed intercalate matrices. We introduce a constrain propagation algorithm which produces a consistently signed intercalate matrix of a given type, or shows that such a matrix does not exist. Thus, we get the corresponding conclusion for sums-of-squares formulas over the integers. We discuss canonical forms for consistently signed intercalate matrices, which significantly improve the efficiency of the algorithm.

\section{Equivalence with Consistently Signed Intercalate Matrices}

In this section, we review the equivalence of sums-of-squares formulas over the integers with consistently signed intercalate matrices.

\begin{defn}
An \emph{intercalate matrix of type $(r,s,n)$} is an $r\times s$ matrix $M$ with entries in the set of $n$ elements $\{1,...,n\}$ such that:
\begin{itemize}
\item The entries along each row are distinct.
\item The entries along each column are distinct.
\item If $M_{ij}= M_{i'j'}$, then $M_{ij'} = M_{i'j}$. (Equivalently, each $2\times 2$ submatrix contains an even number of distinct elements.)
\end{itemize}
Such a matrix is called \emph{consistently signed} if we can assign a sign ($+$ or $-$) to each entry such that if $M_{ij}= M_{i'j'}$, and so $M_{ij'} = M_{i'j}$, then the $2\times 2$ submatrix consisting of these four elements has an odd number of minus signs.
\end{defn}

An example of a consistently signed intercalate matrix of type $(3,5,7)$ is 
\[
\left(\begin{array}{ccccc}
1 & 2 & 3 & 4 & 5\\
2 & -1 & 4 & -3 & 6\\
3 & -4 & -1 & 2 & 7
\end{array}\right).
\]

A consistently signed intercalate matrix $M$ of type $(r,s,n)$ is equivalent to a sums-of-squares formula of type $[r,s,n]$ through the following correspondence:

\begin{itemize}
\item[] If $(i,j)$th entry of $M$ is $\pm k$, then $x_iy_j$ occurs in the expansion of $z_k$ with matching sign.
\end{itemize}

For example, the sums-of-squares formula of type $[r,s,n]$ corresponding to the consistently signed intercalate matrix above is given by
\begin{align*}
z_1 &= x_1y_1 - x_2y_2-x_3y_3,\\
z_2 &= x_1y_2 + x_2y_1 + x_3y_4,\\
z_3 &= x_1y_3 - x_2y_4 + x_3y_1,\\
z_4 &= x_1y_4 +x_2y_3 - x_3y_2,\\
z_5 &= x_1y_5,\\
z_6 &= x_2y_5,\\
z_7 &= x_3y_5.
\end{align*}

\section{Constraint Propagation Algorithm} \label{example}

We now summarize the methods of the constraint propagation algorithm for producing a consistently signed intercalate matrix of a given type $[r,s,n]$ (or concluding that such a formula does not exist). A more detailed treatment of the key functions of the algorithm is included in \ref{detail}, with pseudocode.

The input for this procedure will be a partially completed matrix. For example, for a $[4,4,4]$ formula, we might start with the matrix
\[
\left(\begin{array}{cccc}
1&2&3&4\\
*&*&*&*\\
*&*&*&*\\
*&*&*&*\\
\end{array}\right),
\]
where the asterisks, $*$, indicate that an entry is not yet determined. For large types, choosing a good input is critical to the efficiency of the algorithm. This choice is discussed in \ref{input}.

The output of the procedure is a completed consistently signed intercalate matrix with the form of the input matrix, or the conclusion that such a matrix does not exist.

The procedure works by maintaining a list of possible values at each unknown entry, and updating these lists as entries are filled in. For example, for the input above, the lists would be initialized as
\[
\left(\begin{array}{cccc}
1&2&3&4\\
\{\pm 2, \pm 3, \pm 4\}&\{\pm 1, \pm 3, \pm 4\}&\{\pm 1, \pm 2, \pm 4\}&\{\pm 1, \pm 2, \pm 3\}\\
\{\pm 2, \pm 3, \pm 4\}&\{\pm 1, \pm 3, \pm 4\}&\{\pm 1, \pm 2, \pm 4\}&\{\pm 1, \pm 2, \pm 3\}\\
\{\pm 2, \pm 3, \pm 4\}&\{\pm 1, \pm 3, \pm 4\}&\{\pm 1, \pm 2, \pm 4\}&\{\pm 1, \pm 2, \pm 3\}\\
\end{array}\right),
\]
since the colors of entries in the same column must be distinct.

The next step is to choose a test value from the list of possibilities for one of the entries. The procedure makes a copy of the matrix including the chosen test entry, propagates the consequences of the choice, and continues to work with this matrix. This test case will eventually either result in a complete consistently signed intercalate matrix, or the conclusion that none exists. If it produces a matrix, that matrix is the output for the original procedure. If there is no such matrix, the choice of a test value is eliminated from the original matrix, and the procedure chooses a new possible value to test.

The choice of the test entry can significantly affect the efficiency of the algorithm, and when we eliminate a test entry, we can often reach additional conclusions and eliminate additional possibilities. These refinements are discussed in \ref{testvalues}.

In our example, the procedure might choose to try $2$ as a possibility for the entry $(2,1)$. The procedure makes a copy of the matrix, enters the chosen test value, and propagates the consequences of this test value. It would first eliminate the color $2$ from the remaining entries in the second row and first column, 
\[
\left(\begin{array}{cccc}
1&2&3&4\\
2&\{\pm 1, \pm 3, \pm 4\}&\{\pm 1, \pm 4\}&\{\pm 1, \pm 3\}\\
\{\pm 3, \pm 4\}&\{\pm 1, \pm 3, \pm 4\}&\{\pm 1, \pm 2, \pm 4\}&\{\pm 1, \pm 2, \pm 3\}\\
\{\pm 3, \pm 4\}&\{\pm 1, \pm 3, \pm 4\}&\{\pm 1, \pm 2, \pm 4\}&\{\pm 1, \pm 2, \pm 3\}\\
\end{array}\right).
\]
In order for the $2\times 2$ submatrix in the upper left corner to satisfy the condition on squares, we can also conclude that the $(2,2)$ entry must be $-1$. We then propagate the consequences of this assignment.
\[
\left(\begin{array}{cccc}
1&2&3&4\\
2&-1&\{\pm 4\}&\{\pm 3\}\\
\{\pm 3, \pm 4\}&\{\pm 3, \pm 4\}&\{\pm 1, \pm 2, \pm 4\}&\{\pm 1, \pm 2, \pm 3\}\\
\{\pm 3, \pm 4\}&\{\pm 3, \pm 4\}&\{\pm 1, \pm 2, \pm 4\}&\{\pm 1, \pm 2, \pm 3\}\\
\end{array}\right).
\]

Since we know that the entry $(2,3)$ must be the color $4$ (although we don't know the sign yet), we can eliminate $4$'s from the other entries in the column, and similarly for $3$'s in the last column.

\[
\left(\begin{array}{cccc}
1&2&3&4\\
2&-1&\{\pm 4\}&\{\pm 3\}\\
\{\pm 3, \pm 4\}&\{\pm 3, \pm 4\}&\{\pm 1, \pm 2\}&\{\pm 1, \pm 2\}\\
\{\pm 3, \pm 4\}&\{\pm 3, \pm 4\}&\{\pm 1, \pm 2\}&\{\pm 1, \pm 2\}\\
\end{array}\right).
\]

This process continues, making choices and propagating the consequences in rows, columns, and squares. For our example, we obtain the following matrix by making the choices that the $(2,3)$ entry is $4$, the $(3,1)$ entry is $3$, and the $(4,1)$ entry is $4$.
\[
\left(\begin{array}{cccc}
1&2&3&4\\
2&-1&4&-3\\
3&-4&-1&2\\
4&3&-2&-1\\
\end{array}\right).
\]

\section{More Details} \label{detail}

We now give a detailed description of the various functions used in the constraint propagation algorithm, finishing with high-level pseudo-code for the algorithm. The code itself is available at \url{www.math.umn.edu/~mklynn/sos_pub}.

We begin with the function \textsc{makeMatrix} that takes values for $r$, $s$, and $n$, and sets up an $r\times s$ matrix, where the entries are lists of all possible entries, $[-n,...,-2,-1,1,2,...,n]$. The object that we return includes this matrix, as well as various other structures that will be useful later in the algorithm. These structures include:
\begin{itemize}
\item A list of coordinates for the matrix
\item A list of the rows of the matrix
\item A list of the columns of the matrix
\item A list of the squares of the matrix
\item For each coordinate, a list of all other coordinates in the same row or same column.
\item For each coordinate, a list of all $2\times 2$ submatrices (squares) which include that coordinate.
\item An initially empty list of assignments to propagate
\item An initially empty list of known colors to propagate
\end{itemize}

This provides the framework that we'll use to produce a consistently signed intercalate matrix. As we assign entries, we'll propagate the consequences of this assignment, eliminating possibilities from related entries. By eliminating entries in this way, we greatly reduce the search space compared to a brute force search, or even a back-tracking approach as in \cite{Lynn}.

We now turn our attention to the function which assigns values to entries. This function could be called by a user, who wants to test for specific consistently signed intercalate matrices. It will be called once values become known as a result of constraint propagation, and it will be called to try strategically chosen test values once all constraints have been propagated. Calling this function as a result of constraint propagation will be discussed later in this section. The user's choice of assignment will be discussed in \ref{input}. Choosing good test values will be discussed in \ref{testvalues}.

The function \textsc{assign} takes our matrix as an input, along with a $value$ and a $coordinate$ to which it will be assigned. If the coordinate has already been assigned to be this value, \textsc{assign} immediately returns $True$. Then, if the value is not a possibility for that coordinate, \textsc{assign} returns $False$. Otherwise, \textsc{assign} makes the assignment, adds it to the list of assignments to propagate, and returns $True$.

\begin{algorithm}[H]
\caption{Given a matrix and a value to assign to a coordinate, make this assignment if it's possible.}
\begin{algorithmic}
\Function{assign}{$ourMatrix$, $coordinate$, $value$}
\If{$coordinate$ has already been assigned to be $value$}
\State\Return{True}
\EndIf
\If{$value$ is not a possibility for the entry at $coordinate$}
\State \Return{False}
\EndIf
\State{$ourMatrix[coordinate] \gets value$}
\State{add $(coordinate, value)$ to list of assignments to propagate}
\State \Return True
\EndFunction
\end{algorithmic}
\end{algorithm}

Before we introduce the propagation functions, we introduce the function \textsc{eliminate}, which will be used by the propagation functions. This function takes the our matrix with a $coordinate$ and $value$, and deletes that value as a possibility for that coordinate. If the elimination means that there is only one remaining possibility for that entry, that value is assigned. When the elimination function reduces the entry to a single color with either sign as an option, the color and entry are added to the running list of colors to propagate. The function returns $False$ if the elimination is impossible, and $True$ otherwise.

\begin{algorithm}[H]
\caption{Given a matrix and a value to eliminate from a coordinate, make this elimination if it's possible.}
\begin{algorithmic}
\Function{eliminate}{$ourMatrix$, $coordinate$, $value$}
\If{the value at coordinate is already known}
\If{the known value is different from $value$}
\State \Return{True}
\Else
\State \Return{False}
\EndIf
\EndIf
\If{$value$ is not in the list $ourMatrix[coordinate]$}
\State\Return{True}
\EndIf
\State{delete $value$ from the list $ourMatrix[coordinate]$}
\If{there is only one value, $lastValue$, in $ourMatrix[coordinate]$}
\State{\Return \Call{assign}{$ourMatrix$, $coordinate$, $lastValue$}}
\EndIf
\If{$ourMatrix[coordinate]$ consists of two values of the same color}
\State{add $(coordinate, color)$ to the list of colors to propagate}
\EndIf
\State \Return{True}
\EndFunction
\end{algorithmic}
\end{algorithm}

Note that the structure of the functions \textsc{assign} and \textsc{eliminate} ensure that we will never have duplicate entries in our lists of assignments and colors to propagate.

We now turn our attention to the propagation functions, beginning with the function \textsc{propagateRowsAndColumns}, which takes our matrix, a $coordinate$ and its $value$. The function propagates the consequences of this assignment to the row and column of $coordinate$, eliminating the color of $value$ from all of those entries. The function returns $False$ if one if these eliminations is impossible, and it returns $True$ otherwise.

\begin{algorithm}[H]
\caption{Given a matrix and a value assigned to a coordinate, propagates this assignment to the row and column of coordinate.}
\begin{algorithmic}
\Function{propagateRowsAndColumns}{$ourMatrix$, $coordinate$, $value$}
\For{each $entry$ in the same row or column as $coordinate$}
\If{not \Call{eliminate}{$ourMatrix$, $entry$, $value$}}
\State \Return False
\EndIf
\If{not \Call{eliminate}{$ourMatrix$, $entry$, $-value$}}
\State \Return False
\EndIf
\EndFor
\State \Return True
\EndFunction
\end{algorithmic}
\end{algorithm}

Implementation of the \textsc{propagateSquares} function involves many helper functions, the details of which we omit for brevity. We provide only very high level pseudo-code for the function \textsc{propagateSquare}, which propagates constraints in a single square. The function \textsc{propagateSquares} then calls \textsc{propagateSquare} for each square containing our coordinate.

The function \textsc{propagateSquare} takes our matrix with a $coordinate$, $value$, and a $square$ containing $coordinate$. The function propagates the consequences of this assignment to the given square. It returns $False$ if it encounters a contradiction at some point, otherwise it returns $True$.

\begin{algorithm}[H]
\caption{Given a matrix, a value assigned to a coordinate, and a $2\times 2$ submatrix including this coordinate, propagates this assignment to the given square.}
\begin{algorithmic}
\Function{propagateSquare}{$ourMatrix$, $coordinate$, $value$, $square$}
\State{label the entries of $square$ as $opposite$, $sameRow$, and $sameColumn$,}
\State{based on their relative position to $coordinate$.}
\If{$opposite$ is $value$}
\State{$sameRow$ and $sameColumn$ must be the same color in opposite signs}
\State \Return
\EndIf
\algstore{myalg}
    \end{algorithmic}
    \end{algorithm}
    \begin{algorithm}[H]                 
    \begin{algorithmic}           
    \algrestore{myalg}
\If{$opposite$ is $-value$}
\State{$sameRow$ and $sameColumn$ must have the same value}
\State \Return
\EndIf
\If{$opposite$ is the same color as $value$}
\State{$sameRow$ and $sameColumn$ must have the same color}
\EndIf
\If{$opposite$ is a different color from $value$}
\State{$sameRow$ and $sameColumn$ must have different colors}
\EndIf
\If{$sameRow$ and $sameColumn$ are the same value}
\State{\Return\Call{assign}{$ourMatrix$, $opposite$, $-value$}}
\EndIf
\If{$sameRow$ and $sameColumn$ are the same color in opposite signs}
\State{\Return\Call{assign}{$ourMatrix$, $opposite$, $value$}}
\EndIf
\If{$sameRow$ and $sameColumn$ are the same color}
\State{$opposite$ must be $\pm value$}
\State \Return
\EndIf
\If{$sameRow$ and $sameColumn$ are different colors}
\State{\Return (\Call{eliminate}{$ourMatrix$, $opposite$, $value$} and \\\hspace{2cm} \Call{eliminate}{$ourMatrix$, $opposite$, $-value$})}
\EndIf
\If{$value$ is not a possibility for $opposite$}
\State{$sameRow$ and $sameColumn$ can't have the same color in opposite}
\State{signs}
\EndIf
\If{$-value$ is not a possibility for $opposite$}
\State{$sameRow$ and $sameColumn$ can't have the same value}
\EndIf
\State \Return True
\EndFunction
\end{algorithmic}
\end{algorithm}

All of the above functions rely on knowing the value of one entry, and determining how this affects the values of related entries. However, in the example in \ref{example}, we saw that we were able to make similar eliminations when only the color (but not the sign) of an entry is known. For propagating to rows and columns, we can still use the \textsc{propagateRowsAndColumns} function, since the sign of the entry has no influence on this function. However, for the squares, we need a separate function, which we now define.

The function \textsc{propagateSquaresColor} takes our matrix with a $coordinate$ and $color$, and it propagates the consequences of knowing that the value of $coordinate$ is $\pm color$ to all $2\times 2$ submatrices containing $coordinate$. It returns $False$ if it encounters a contradiction at some point, otherwise it returns $True$.

\begin{algorithm}[H]
\caption{Given a matrix and the color of a coordinate, propagates this constraint to all $2\times 2$ submatrices which include this coordinate.}
\begin{algorithmic}
\Function{propagateSquaresColor}{$ourMatrix$, $coordinate$, $color$}
\For{each $2\times 2$ submatrix of $ourMatrix$ including $coordinate$}
\State{label the entries of the square as $opposite$, $sameRow$, and $sameColumn$,}
\State{based on their relative position to $coordinate$.}
\If{$opposite$ is known to be the same color as $color$}
\State{$sameRow$ and $sameColumn$ must be the same color}
\State \Return
\EndIf
\If{$opposite$ is known to be a different color from $color$}
\State{$sameRow$ and $sameColumn$ must have different colors}
\State \Return
\EndIf
\If{$sameRow$ and $sameColumn$ are the same color}
\State{$opposite$ must have the same color as $coordinate$}
\EndIf
\If{$sameRow$ and $sameColumn$ can't be the same color}
\State{$opposite$ must have a different color from $coordinate$}
\EndIf
\EndFor
\EndFunction
\end{algorithmic}
\end{algorithm}

We now give psuedo-code for our overall propagation function, \textsc{propagate}, which calls the above functions as long as there are values and colors to propagate. 

\begin{algorithm}[H]
\caption{Given a matrix, performs all propagations until there are no more values or colors to propagate. If there is a contradiction at some point, it immediately returns $False$. If it doesn't produce a contradiction, return $True$.}
\begin{algorithmic}
\Function{propagate}{$matrix$}
\While{there are value pairs $(entry, value)$ to be propagated}
\If{not \Call{propagateRowsAndColumns}{$matrix$, $entry$, $value$}}
\State\Return{$False$}
\EndIf
\If{not \Call{propagateSquares}{$matrix$, $entry$, $value$}}
\State\Return{$False$}
\EndIf
\While{there are color pairs $(entry, color)$ to be propagated}
\If{not \Call{propagateRowsAndColumns}{$matrix$, $entry$, $color$}}
\State\Return{$False$}
\EndIf
\If{not \Call{propagateSquaresColor}{$matrix$, $entry$, $color$}}
\State\Return{$False$}
\EndIf
\EndWhile
\EndWhile
\Return{$True$}
\EndFunction
\end{algorithmic}
\end{algorithm}

Before we describe the overall algorithm for finding consistently signed intercalate matrices, we describe two verification functions.

The first verification function, \textsc{verify}, checks there is no contradiction in the values assigned so far in the matrix, and ignores the unknown values. It returns $True$ if there is no contradiction so far, and $False$ if there is a contradiction.

\begin{algorithm}[H]
\caption{Given a matrix, verifies that there is no contradiction among assigned values}
\begin{algorithmic}
\Function{verify}{$ourMatrix$}
\For{each row of $ourMatrix$}
\If{there are two known entries with the same color}
\State \Return{$False$}
\EndIf
\EndFor
\For{each column of $ourMatrix$}
\If{there are two known entries with the same color}
\State \Return{$False$}
\EndIf
\EndFor
\For{each $2\times 2$ submatrix of $ourMatrix$}
\If{all entries are known}
\State{identify the relative locations of the entries as $coordinate$, $opposite$,}
\State{$sameRow$, and $sameColumn$.}
\If{not (($coordinate = opposite$ and $sameRow = -sameColumn$) or \\($coordinate = -opposite$ and $sameRow = sameColumn$) or \\($coordinate \neq opposite$ and $coordinate \neq -opposite$ and \\ $sameRow \neq sameColumn$ and $sameRow \neq -sameColumn$))}
\State\Return{$False$}
\EndIf
\EndIf
\EndFor
\State\Return{$True$}
\EndFunction
\end{algorithmic}
\end{algorithm}

The second verification function, \textsc{verifyComplete}, checks if every entry of the matrix has been assigned a value. It returns $True$ if the matrix is complete, and it returns $False$ otherwise.

We now provide the pseudocode for the main algorithm, which takes a partially completed matrix of type $[r,s,n]$ as input, and either returns a consistently signed intercalate matrix of the given type, or returns $False$ if such a matrix does not exist. This function, $completeMatrix$, is called recursively in order to find the desired matrix.

Prior to calling the function $completeMatrix$, a user can choose to make some assignments. Choosing these assignments wisely can significantly reduce the run time of the algorithm, while maintaining correctness. Choice of these initial assignments is discussed in \ref{input}.

The algorithm works by choosing a value to test at a particular input, and propagating the consequences of that assignment. If the propagation does not lead to a contradiction and does not complete the matrix, this process is repeated. If the propagation leads to a completed consistently signed intercalate matrix, that matrix is returned. If the propagation leads to a contradiction, we backtrack and eliminate that value as a possibility from the entry where it was tested. A new test value is then chosen, and the process repeats.

In some situations, when we backtrack and eliminate a test value, we can make other eliminations as well. The choice of the test value and these eliminations are discussed in \ref{testvalues}.

\begin{algorithm}[H]
\caption{Given a matrix, returns a consistently signed intercalate matrix with the missing values filled in, or returns false if no such matrix exists.}
\begin{algorithmic}
\Function{completeMatrix}{$matrix$}
\If{not \Call{verify}{$matrix$}}
\State\Return{False}
\EndIf
\If{\Call{verifyComplete}{$matrix$}}
\State\Return{$matrix$}
\EndIf
\If{there is a choice for $testValue$ and $coordinate$}
\State{$copyMatrix\gets matrix$}
\State\Call{assign}{$copyMatrix$, $coordinate$, $testValue$}
\If{not \Call{propagate}{$copyMatrix$}}
\If{not \Call{eliminate}{$matrix$, $coordinate$, $testValue$}}
\State\Return{$False$}
\EndIf
\State\Return\Call{completeMatrix}{$matrix$}
\EndIf
\If{not \Call{completeMatrix}{$copyMatrix$}}
\State{\Call{eliminate}{$matrix$, $coordinate$, $testValue$}}
\State\Return\Call{completeMatrix}{$matrix$}
\Else
\State\Return\Call{completeMatrix}{$copyMatrix$}
\EndIf
\EndIf
\EndFunction
\end{algorithmic}
\end{algorithm}

An implementation of these functions in python is available at \url{www.math.umn.edu/~mklynn/sos_pub}, with a sample use of these functions.

In this implementation, we include two different versions of the complete matrix function, corresponding to two different methods for selecting test values. These methods are discussed in \ref{testvalues}.

\section{Group Action and Choice of Input} \label{input}

If we have a consistently signed intercalate matrix $M$ of type $(r,s,n)$, we can produce many more such matrices by manipulating our matrix. In particular, we can obtain another consistently signed intercalate matrix with the following operations:
\begin{itemize}
\item permuting the rows of $M$,
\item permuting the columns of $M$,
\item permuting the colors of $M$,
\item flipping the signs of all elements in a row (or multiple rows) of $M$,
\item flipping the signs of all elements in a column (or multiple columns) of $M$,
\item flipping the signs of all elements of a color (or multiple colors) of $M$.
\end{itemize}
We can also view these operations as an action of $O(r,\mathbb{Z})\times O(s,\mathbb{Z})\times O(n,\mathbb{Z})$ on the set of sums-of-squares formulas of type $(r,s,n)$ over $\mathbb{Z}$, as in \cite{Lynn}.

When searching for formulas using our algorithm, we can use this fact to significantly reduce the number of potential matrices that we test. In particular, we can choose our input based on these observations.

\begin{prop}
There exists a consistently signed intercalate matrix of type $(r,s,n)$ if and only if there is a consistently signed intercalate matrix of type $(r,s,n)$ such that
\begin{itemize}
\item The first row has entries $1,2,3,....,s$.
\item The $(i,i)$th entry is $1$ for $i\leq \left\lceil \frac{rs}{n}\right\rceil$.
\end{itemize}
\end{prop}

\begin{proof}
Suppose we have a consistently signed intercalate matrix $M$ of type $(r,s,n)$.

For an $r\times s$ matrix with entries chosen from $n$ colors, by the pigeonhole principle, some color, $x$, must occur at least $N = i\leq \left\lceil \frac{rs}{n}\right\rceil$ times. Let $M'$ be the matrix obtained from $M$ by swapping the colors $1$ and $x$, so that the color $1$ occurs $N$ times.

Suppose $\pm 1$ occurs in the entry $(x_1,y_1)$. If we swap the first and $x_1$th rows and swap the first and $y_1$th columns, that $\pm 1$ is now in the entry $(1,1)$. Assuming $N>1$, there is another $\pm 1$ in the matrix. Since the colors along rows and columns must be distinct, it occurs at an entry $(x_2,y_2)$, where $x_2 \neq 1$ and $y_2 \neq 1$. If we swap the second and $x_2$th rows and swap the second and $y_2$th columns, then that $\pm 1$ is now in the entry $(2,2)$. Continuing this, we get $\pm 1$ in the entries $(i,i)$ for $i\leq N$.

Now, for $i$ such that the entry $(i,i)$ is $-1$, we flip the signs in the $i$th row. This ensures that the entries $(i,i)$ are all $1$ for $i\leq N$.

The first row now has entries $1,z_2,...$. Since colors along rows must be distinct, $z_2\neq \pm 1$. We then swap the color of $z_2$ and $2$, so the second entry is now $\pm 2$. If it is $-2$, we flip the signs of all $2$'s in the matrix. The first row then has entries $1,2,z_3,...$, where the color of $z_3$ is not $1$ or $2$. We swap this color with $3$, and fix the sign if needed. Continuing this for the rest of the entries in the row, we obtain a matrix with entries $1,2,3,...,s$ along the first row.

Note that the operations on the first row left the $1$'s along the diagonal intact, so we have a matrix of the desired form.
\end{proof}

With a matrix of this form, we also know that the first column must have $1,-2,-3,...,-N$ for the first $N$ entries, due to the conditions on squares. However, the propagation algorithm immediately fills in these values, so we don't usually bother including them with our input.

More careful analysis of the form of consistently signed intercalate matrices of particular types can yield additional assumptions about our input matrix.

\section{Choice of Test Values} \label{testvalues}

The algorithm for producing a consistently signed intercalate matrix necessitates making a choice of a test coordinate and value, perhaps many times. Different selections for test values can have a dramatic effect on the run time of the algorithm, so we want to make these choices carefully.

In the implementation of the algorithm available at \url{www.math.umn.edu/~mklynn/sos_pub}, we provide two different versions of the function \textsc{completeMatrix}, which make different choices for these test values.

In the first version, we choose the entry with the fewest possibilities, whose value is not yet known. We select the first of the possible values, and test that value. The idea behind this choice is that with fewer possibilities at that entry, we'll have fewer test values to cycle through, so this should help make the algorithm faster, compared to testing a coordinate with more possible values.

In the second version, we choose our test value based on combinatorial arguments about the frequency of each color in the matrix. This often involves using knowledge about smaller types of consistently signed intercalate matrices. As an example, suppose we would like to know if there exists a sums-of-squares formula of type $[10,17, 28]$. In fact, whether or not a formula of this type exists remains an open question. In refining our search for this formula, we can make use of the following observations.

\begin{prop}
If there is a consistently signed matrix of type $(10,17,28)$, then:
\begin{enumerate}
\item Some color must occur at least $7$ times.
\item Each color must occur at least $3$ times.
\end{enumerate}
\end{prop}

\begin{proof}
\begin{enumerate}
\item We have $10\cdot 17 = 170$ entries, and $28$ choices for colors. This means that each color will occur an average of $\left\lfloor \frac{170}{28}\right\rfloor\approx 6.07$. So, some color will need to occur at least $7$ times.
\item For this result, we use the fact that there is no sums-of-squares formula of type $[10,15,27]$.

Suppose some color $c$ occurs $2$ or fewer times in a consistently signed intercalate matrix of type $[10,17,28]$. Now, imagine we delete the $2$ (or fewer) columns where this color occurs. This results in a $10\times 15$ matrix, which has no entries with color $c$. So, the number of colors in this matrix is at most $27$, and the matrix is a consistently signed intercalate matrix of type $[10,15,27]$. But such a matrix cannot exist, so we have a contradiction, and every color must occur at least $3$ times.
\end{enumerate}
\end{proof}

By permuting the colors, we may assume that the colors are in decreasing order of frequency. That is, $1$ is the most prevalent color, and $n$ is the least prevalent color. We can use information we deduce about the frequency of each color to form a minimum signature.

\begin{defn}
A \emph{minimum signature} for a consistently signed intercalate matrix of type $(r,s,n)$ is a function $m:\{1,2,...n\}\rightarrow\mathbb{Z}$ such that the number of times the color $x$ appears in the matrix is at least $m(x)$.
\end{defn}

So, a minimum signature provides a lower bound for the frequency of each color. In the second version of \textsc{completeMatrix}, we first focus on meeting these minimum frequencies. Given a minimum signature $m$, the function finds the first color $x$ such that fewer than $m(x)$ entries have been assigned color $x$. The function finds the first entry where $x$ is still a possibility, and that is chosen as the test value and coordinate. If there is no entry where $x$ is still a possibility, there is no consistently signed intercalate matrix with the desired form and minimum signature, and so the function will return $False$. Once there is no color $x$ with frequency less than $m(x)$, we revert to the method for choosing test values from the original \textsc{completeMatrix} function: choosing the first value in the entry with the fewest possibilities.

Comparisons of the run times for these versions is included in \ref{runtimes}.

\section{Run Times}	\label{runtimes}

\begin{wrapfigure}{r}{0.4\textwidth}{\footnotesize
$\begin{array}{|c|c|c|c|}
\hline
\textrm{Type} & \textrm{v. 1} & \textrm{v. 2}& \textrm{brute}\\
\hline
(2,3,3) &<.01s&<.01s&.19s\\
(2,5,5) &<.01s&<.01s&-\\
(2,7,7) &.04s&.04s&-\\
(2,9,9) &.79s&.80s&-\\
(3,3,3) &<.01s&<.01s&32s\\
(3,5,6) &.01s&.01s&-\\
(3,6,7) &.05s&.09s&-\\
(3,7,7) &.07s&.07s&-\\
(3,9,10) &45s&61s&-\\
(4,5,7) &.57s&.91s&-\\
(4,6,7) &.08s&.18s&-\\
(4,7,7) &.10s&.11s&-\\
(4,9,11) &425s&650s^*&-\\
(5,5,7) &.43s&.26s&-\\
(5,6,7) &.11s&.29s&-\\
(5,7,7) &.15s&.14s&-\\
(6,6,7) &.16s&.42s&-\\
(6,7,7) &.20s&.20s&-\\
(7,7,7) &.26s&.26s&-\\
\hline
\end{array}$}
\end{wrapfigure}

In this section, we compare run times of our two versions of the function \textsc{completeMatrix} with a brute force search. The results are recorded in the following table. Each time is the average over 10 runs of the program, on MacBook Pro with 3.1 GHz Intel Core i7 processor and 16 GB 1867 MHz DDR3, except for times marked with an asterisk, which were run only once. 

For these tests, our initial assignments were to fill in the first row as $1,2,3,...,r$ and to fill in $\left\lceil\frac{rs}{n}\right\rceil$ $1$'s along the diagonal, as discussed in \ref{input}. For version 2, we require that every color occur at least once; note that this is suboptimal, as discussed in \ref{testvalues}. We run these tests for several types $(r,s,n)$ where $r\leq s\leq 9$.

On the right, we have a table of run times for types $(r,s,n)$, where $n$ is the largest integer such that a formula of type $(r,s,n)$ does not exist. Below, we give a table of run times for types $(r,s,n)$, where $n$ is the smallest value such that a formula of type $(r,s,n)$ exists.

{\footnotesize
\begin{minipage}{0.5\textwidth}
\[
\begin{array}{|c|c|c|c|}
\hline
\textrm{Type} & \textrm{v. 1} & \textrm{v. 2}& \textrm{brute}\\
\hline
(2,3,4) &<.01s&<.01s&1.50s\\
(2,4,4) &<.01s&<.01s&60s\\
(2,5,6) &<.01s&<.01s&-\\
(2,6,6) &<.01s&<.01s&-\\
(2,7,8) &<.01s&<.01s&-\\
(2,8,8) &<.01s&<.01s&-\\
(2,9,10) &<.01s&<.01s&-\\
(3,3,4) &<.01s&<.01s&632s\\
(3,4,4) &<.01s&<.01s&-\\
(3,5,7) &.02s&.02s&-\\
(3,6,8) &<.01s&<.01s&-\\
(3,7,8) &<.01s&.01s&-\\
(3,8,8) &.01s&.01s&-\\
(3,9,11) &15s&6.99s&-\\
(4,4,4) &<.01s&<.01s&-\\
(4,5,8) &.01s&.01s&-\\
(4,6,8) &.02s&.01s&-\\
\hline
\end{array}
\]
\end{minipage}
\begin{minipage}{0.5\textwidth}
\[
\begin{array}{|c|c|c|c|}
\hline
\textrm{Type} & \textrm{v. 1} & \textrm{v. 2}& \textrm{brute}\\
\hline
(4,7,8) &.02s&.02s&-\\
(4,8,8) &.02s&.02s&-\\
(4,9,12) &.04s&.04s&-\\
(5,5,8) &.02s&.02s&-\\
(5,6,8) &.02s&.02s&-\\
(5,7,8) &.02s&.03s&-\\
(5,8,8) &.03s&.03s&-\\
(6,6,8) &.02s&.03s&-\\
(6,7,8) &.03s&.03s&-\\
(6,8,8) &.03s&.03s&-\\
(7,7,8) &.04s&.04s&-\\
(7,8,8) &.04s&.04s&-\\
(7,9,15) &-&.31s&-\\
(8,8,8) &.05s&.05s&-\\
(8,9,16) &.23s&.20s&-\\
(9,9,16) &.21s&.23s&-\\
\hline
\end{array}
\]
\end{minipage}
}

For many of the test cases that we run, the run times for the two versions of the algorithm are very close. Recall that, in the second version of the algorithm, once the minimum frequencies are met, it reverts to choosing test values based on the heuristic of the first version. For small types, the minimum frequencies are often achieved with the initial assignments, so the two versions are functionally identical. For larger types such as $(4,9,11)$ and $(3,9,11)$, we observe clear differences in the run times. However, neither version consistently produces faster running times. Future improvements to the algorithm could be made by exploring different ways to choose test values, and how they affect the run times.

In an appendix, we include a table of known values of $r*_{\mathbb{Z}}s$, the smallest value of $n$ such that a sums-of-squares formula of type $[r,s,n]$ exists over the integers. In many cases, the exact value of $r*_{\mathbb{Z}}s$ is not known, and we instead include an upper bound (indicated with an asterisk). These tables are included as guides for which types $(r,s,n)$ would be interesting to explore using our algorithms.

%\subsection*{Acknowledgments} 

\appendix

\section{Values of $r*_\mathbb{Z}s$}

We include tables of known values of $r*_\mathbb{Z}s$, the smallest number $n$ such that there is a sums-of-squares formula of type $[r,s,n]$ over the integers. In the cases where only an upper bound is known, we include that upper bound with an asterisk. These tables are compiled from those in \cite{ShapiroBook}.

\vspace{.5cm}

\begin{center}
$r*_\mathbb{Z}s$ for $1\leq r,s\leq 17$
\end{center}

{\small
\begin{tabular}{ c | ccccccccccccccccc }
$r*_\mathbb{Z}s$ & 1&2&3&4&5&6&7&8&9&10&11&12&13&14&15&16&17\\
\hline
1&1&2&3&4&5&6&7&8&9&10&11&12&13&14&15&16&17\\
2&&2&4&4&6&6&8&8&10&10&12&12&14&14&16&16&18\\
3&&&4&4&7&8&8&8&11&12&12&12&15&16&16&16&19\\
4&&&&4&8&8&8&8&12&12&12&12&16&16&16&16&20\\
5&&&&&8&8&8&8&13&14&15&16&16&16&16&16&21\\
6&&&&&&8&8&8&14&14&16&16&16&16&16&16&22\\
7&&&&&&&8&8&15&16&16&16&16&16&16&16&23\\
8&&&&&&&&8&16&16&16&16&16&16&16&16&24\\
9&&&&&&&&&16&16&16&16&16&16&16&16&25\\
10&&&&&&&&&&16&26&26&27&27&28&28&$29^*$\\
11&&&&&&&&&&&26&26&28&28&30&30&$32^*$\\
12&&&&&&&&&&&&26&28&30&32&32&32\\
13&&&&&&&&&&&&&28&32&32&32&32\\
14&&&&&&&&&&&&&&32&32&32&32\\
15&&&&&&&&&&&&&&&32&32&32\\
16&&&&&&&&&&&&&&&&32&32\\
17&&&&&&&&&&&&&&&&&32
\end{tabular}
}

\vspace{.5cm}

\begin{center}
$r*_\mathbb{Z}s$ for $1\leq r \leq 30$ and $18\leq s\leq 30$
\end{center}

{\small
\begin{tabular}{ c |  ccccccccccccc}
$r*_\mathbb{Z}s$ &18&19&20&21&22&23&24&25&26&27&28&29&30\\
\hline
1&18&19&20&21&22&23&24&25&26&27&28&29&30\\
2&18&20&20&22&22&24&24&26&26&28&28&30&30\\
3&20&20&20&23&24&24&24&27&28&28&28&31&32\\
4&20&20&20&24&24&24&24&28&28&28&28&32&32\\
5&22&23&24&24&24&24&24&29&30&31&32&32&32\\
6&22&24&24&24&24&24&24&30&30&32&32&32&32\\
7&24&24&24&24&24&24&24&31&32&32&32&32&32\\
8&24&24&24&24&24&24&24&32&32&32&32&32&32\\
9&26&27&28&29&30&31&32&32&32&32&32&32&32\\
10&$29^*$&$30^*$&$30^*$&30&30&32&32&32&32&32&32&32&32\\
11&$32^*$&$32^*$&$32^*$&$42^*$&$44^*$&$44^*$&$44^*$&$46^*$&$48^*$&$48^*$&$48^*$&$48^*$&$52^*$\\
12&32&32&32&$42^*$&$44^*$&$44^*$&$44^*$&$48^*$&$48^*$&$48^*$&$48^*$&$48^*$&$52^*$\\
13&32&$43^*$&$44^*$&$44^*$&$44^*$&$48^*$&$48^*$&$48^*$&$48^*$&$48^*$&$58^*$&$58^*$&$58^*$\\
14&32&$43^*$&$44^*$&$46^*$&$48^*$&$48^*$&$48^*$&$48^*$&$48^*$&$48^*$&$58^*$&$58^*$&$58^*$\\
15&32&$44^*$&$46^*$&$48^*$&$48^*$&$48^*$&$48^*$&$48^*$&$48^*$&$48^*$&$60^*$&$62^*$&$63^*$\\
16&32&$44^*$&$46^*$&$48^*$&$48^*$&$48^*$&$48^*$&$48^*$&$48^*$&$48^*$&$60^*$&$62^*$&$64^*$\\
17&32&$49^*$&$50^*$&$51^*$&$52^*$&$53^*$&$54^*$&$55^*$&$56^*$&$57^*$&$61^*$&$64^*$&$64^*$\\
18&$50^*$&$50^*$&$52^*$&$52^*$&$54^*$&$54^*$&$56^*$&$56^*$&$57^*$&$57^*$&$64^*$&$64^*$&$64^*$\\
19&&$56^*$&$56^*$&$59^*$&$60^*$&$60^*$&$64^*$&$64^*$&$64^*$&$64^*$&$64^*$&$64^*$&$64^*$\\
20&&&$56^*$&$60^*$&$60^*$&$60^*$&$64^*$&$64^*$&$64^*$&$64^*$&$64^*$&$64^*$&$64^*$\\
21&&&&$64^*$&$64^*$&$64^*$&$64^*$&$72^*$&$76^*$&$77^*$&$80^*$&$80^*$&$84^*$\\
22&&&&&$68^*$&$72^*$&$72^*$&$72^*$&$78^*$&$80^*$&$80^*$&$80^*$&$84^*$\\
23&&&&&&$72^*$&$72^*$&$72^*$&$78^*$&$80^*$&$84^*$&$88^*$&$90^*$\\
24&&&&&&&$72^*$&$72^*$&$80^*$&$80^*$&$88^*$&$88^*$&$90^*$\\
25&&&&&&&&$72^*$&$80^*$&$80^*$&$88^*$&$94^*$&$95^*$\\
26&&&&&&&&&$80^*$&$80^*$&$89^*$&$94^*$&$96^*$\\
27&&&&&&&&&&$89^*$&$89^*$&$96^*$&$96^*$\\
28&&&&&&&&&&&$96^*$&$96^*$&$96^*$\\
29&&&&&&&&&&&&$96^*$&$96^*$\\
30&&&&&&&&&&&&&$96^*$\\
\end{tabular}
}

\end{document}